\documentclass[a4paper,11pt]{article}

\usepackage{fullpage}
\usepackage{times}
\usepackage{soul}
\usepackage{url}
\usepackage[utf8]{inputenc}
\usepackage[small]{caption}
\usepackage{graphicx}
\usepackage{xcolor}
\usepackage{amsmath}
\usepackage{booktabs}
\usepackage{algorithm}
\usepackage{algorithmic}
\usepackage{dsfont}
\usepackage{enumitem}
 
\usepackage{amsthm}
\usepackage{amssymb}

\usepackage{bbm}

\newtheorem{theorem}{Theorem}[section]

\newtheorem{proposition}[theorem]{Proposition}
\newtheorem{lemma}[theorem]{Lemma}

\theoremstyle{definition}
\newtheorem{definition}[theorem]{Definition}
\newtheorem{example}[theorem]{Example}

\DeclareMathOperator*{\argmax}{arg\,max}
\DeclareMathOperator*{\argmin}{arg\,min}

\usepackage{todonotes}

\usepackage{natbib}

\allowdisplaybreaks

\usepackage{authblk}

\title{\bf For One and All: Individual and Group Fairness in the \\ Allocation of Indivisible Goods}

\author[1]{Jonathan Scarlett}
\author[2]{Nicholas Teh}
\author[3]{Yair Zick}

\affil[1]{National University of Singapore, Singapore}
\affil[2]{University of Oxford, UK}
\affil[3]{University of Massachusetts Amherst, USA}

\date{\vspace{-10mm}}

\newcommand{\BibTeX}{\rm B\kern-.05em{\sc i\kern-.025em b}\kern-.08em\TeX}
\newcommand{\R}{\mathbb{R}}
\newcommand{\group}{\mathrm{Grp}}

\newcommand{\floor}[1]{\left\lfloor #1 \right\rfloor}
\newcommand{\ceiling}[1]{\left\lceil #1 \right\rceil}

\newcommand{\ef}[1]{$i$-EF#1}
\newcommand{\wef}[1]{$g$-WEF#1}
\newcommand{\mnw}[1]{$i$-MNW}
\newcommand{\mwnw}[1]{$g$-MWNW}

\begin{document}

\maketitle

\begin{abstract}
Fair allocation of indivisible goods is a well-explored problem. 
Traditionally, research focused on \emph{individual fairness} --- are individual agents satisfied with their allotted share? --- and \emph{group fairness} --- are groups of agents treated fairly? 
In this paper, we explore the coexistence of individual envy-freeness (\ef{}) and its group counterpart, group weighted envy-freeness (\wef{}), in the allocation of indivisible goods. 
We propose several polynomial-time algorithms that provably achieve \ef{} and \wef{} simultaneously in various degrees of approximation under three different conditions on the agents' (i) when agents have identical additive valuation functions, \ef{X} and \wef{1} can be achieved simultaneously; 
(ii) when agents within a group share a common valuation function, an allocation satisfying both \ef{1} and \wef{1} exists; and
(iii) when agents' valuations for goods within a group differ, we show that while maintaining \ef{1}, we can achieve a $\frac{1}{3}$-approximation to ex-ante \wef{1}. 
Our results thus provide a first step towards connecting individual and group fairness in the allocation of indivisible goods, in hopes of its useful application to domains requiring the reconciliation of diversity with individual demands.
\end{abstract}


\section{Introduction}

Fairly allocating indivisible goods is a fundamental problem at the intersection of computer science and economics \citep{steinhaus1948fairdivision, alkan1991fairallocation, herve2004fairdivision, brandt2016handbook}.  For instance, a classic problem in fair allocation involves the allocation of courses to students \citep{budish2012harvard, hoshino2014courseallocationquest, budish2016wharton}.  Courses have limited capacity, and slots are often allocated via a centralized mechanism. 
More broadly, analogous problems also surface in numerous areas of social concern, such as the distribution of vaccines to hospitals \citep{pathak2021vaccines}, the allocation of educational resources, and access to infrastructure such as transport, water, and electricity \citep{burmann2020resourceuncertain,wang2003water}.

Several recent works have explored a variety of \emph{distributive justice criteria}; these broadly fall into two categories -- \emph{individual} (e.g., that individual students are not envious of their peers), and \emph{group} (e.g. that students of certain ethnic, gender or professional groups are treated fairly overall). 
While both individual and group fairness have been studied extensively in recent works, to our knowledge, there have been no works proposing algorithms that ensure both concurrently in the setting of indivisible goods with fixed groups. 
In this work, we explore
\begin{quote}
    \emph{efficient algorithms that concurrently ensure approximate individual and group fairness, for certain classes of agent valuation functions.}
\end{quote}
The tension between individual and group fairness exists in a variety of allocation scenarios studied in the literature; for example, when allocating public resources (such as housing, slots in public schools, or scheduling problems in general) \citep{abdulkadiroglu1998housing, abdulkadiroglu2003schoolchoice, abdulkadiroglu2005newyork, abdulkadiroglu2005bostonschool, pathak2011studentassignment, benabbou2018diversity,ElkindKrTe22} --- it is important to maintain fairness towards individual recipients, as well as groups (such as ethnic or socioeconomic groups). Another example is the allocation of reviewers (who, in this metaphor, are the goods) to papers \citep{charlin2013tpms,stelmakh2019peerreview4all}, it is important to balance the individual papers' satisfaction with their allotted reviewers, and the overall quality of reviewers assigned to tracks (e.g. ensuring that the overall reviewer quality for the \textit{Social Choice and Cooperative Game Theory} track is commensurate with that of reviewers for the \textit{Learning and Adaptation} track). 

In this paper, we address the question of whether individually envy-free and group weighted envy-free allocations can co-exist when allocating indivisible goods. 
We present algorithms that compute approximately individually envy-free (EF) and group weighted envy-free (WEF) allocations, where the approximation quality depends on the class of agents' valuation functions.

One could view the WEF property as (subjectively) comparing the \emph{average} (i.e. weighted) utility of an agent within a group to the average utilities of agents in other groups. However, a common flaw of this property is that it is susceptible to outliers: an agent who gets a good with an extremely high bundle value can potentially deprive other group members of valuable items. Hence, by imposing both group and individual fairness, we obtain more equitable outcomes. 

\subsection{Our Contributions}\label{sec:contrib}
We design algorithms that (approximately) reconcile individual and group envy-freeness in the allocation of indivisible goods. The strength of our results naturally depends on the generality of the valuation classes we consider, with more general valuations yielding worse approximation guarantees. 

Our main results and technical analyses are in Section \ref{sec:efwef}. 
In Section \ref{sec:allcommon}, we show that when agents have identical valuation functions, envy-freeness up to any good (EFX) can be achieved in conjunction with group weighted envy-freeness up to one good (WEF1). 
In Section \ref{sec:groupcommon}, when agents within each group have common valuation functions, then envy-freeness up to one good (EF1) can be satisfied together with WEF1. In Section \ref{sec:general}, when valuation functions are distinct, we show that we can obtain a constant factor $\frac{1}{3}$ approximation to WEF1 ex-ante (see Definition \ref{wef1expectation}). 

\subsection{Related Work}\label{sec:related}
Fairly allocating indivisible goods is a fundamental problem at the intersection of computer science and economics \citep{steinhaus1948fairdivision,herve2004fairdivision,brandt2016handbook}. 
 Envy-freeness (EF) is a particularly important individual fairness notion when deciding how to fairly allocate indivisible goods \citep{moulin1995noenvytest,ohseto12006indivisiblegoods}. The existence of approximate EF allocations in conjunction with other individual fairness notions and welfare measures (such as proportionality \citep{aziz2019chores}, pareto-optimality \citep{caragiannis2016mnw}, maximin share \citep{budish2011ef1}) have been studied extensively.

\cite{conitzer2019groupfairness} and \cite{aziz2020groupefchores} introduce the notion of group fairness (applied to every partition of agents within the population), with both offering the ``up to one good'' relaxation of removing one good per agent. A similar concept was also considered in the economics literature~\citep{berliant1992ef}. \cite{benabbou2019groupfairness} explore the relationship between metrics such as utilitarian social welfare in connection with group-wise fairness via an optimization approach.

Several works also suggest notions of group envy-freeness \citep{aleksandrov2018groupef,kyropoulou2019groupallocation,FarhadiGhHa19}; we focus on a popular notion called weighted envy-freeness (WEF) \citep{chakraborty2020wef,ChakrabortyScSu21,ChakrabortySeSu22,montanari2022weightedsubmodular}, which focuses on group fairness with fixed groups, allowing us to study guarantees with the removal of a single good per group. \cite{conitzer2019groupfairness} raised this setting of pre-defined groups as an open question.

\section{Preliminaries}\label{sec:preliminaries}
In the problem of allocating indivisible goods, we are given a set of \emph{agents} $N = \{p_1, \dots, p_n \}$ and \emph{goods} $G= \{g_1, \dots, g_m \}$. Subsets of goods in $G$ are referred to as \emph{bundles}. Agents belong to predefined \emph{groups} (or \emph{types}) $\mathcal{T} = \{T_1,\dots, T_\ell\}$. We assume that $\bigcup_{k=1}^\ell T_k = N$, and that no two groups intersect. Furthermore, each group $T_k$ has a \emph{weight} $w_k$, corresponding to its size, i.e. $w_k = |T_k|$. 
Each agent $p_i\in N$ has a non-negative \emph{valuation function} over bundles of goods: $v_i:2^G \to \R_+$. 
We assume that $v_i$ is additive, which is a common assumption in the fair division literature \citep{caragiannis2016mnw,biswas2018cardinality,conitzer2019groupfairness}, i.e, that $v_i(S) = \sum_{g\in S} v_i(\{g\})$. 
When all agents have the same valuation, we denote their \emph{common valuation} by $v$.

In our framework, we consider the direct allocation of goods to agents, whilst taking into consideration agents' group affiliation, and in the process achieving both individual and group envy-freeness. For example, in the case of assigning reviewers to papers, papers first specify which groups they would like to belong to (by specifying a primary subject area), which implicitly allocates them to a group. Next, these papers are directly allocated reviewers, where their declared subject area is taken into consideration. Thus, the group allocation is not explicitly determined in the allocation process, but is induced from the individual allocations $\mathcal{A} = (A_1, ..., A_n)$ instead. We denote $\group_k(\mathcal{A}) = \bigcup_{i: p_i \in T_k} A_i$ as the induced \textit{group bundle} for $T_k$. To keep our notations simple, for any group $T_k \in \mathcal{T}$, we will let $B_k = \group_k(\mathcal{A})$ denote this induced group bundle. We also let the group utility for $T_k$ be $v_{T_k}(B_k) = \sum_{i:p_i \in T_k} v_i(A_i)$.

Envy-freeness was introduced by \cite{foley1967ef} (see also \cite{brandt2016handbook}, and \cite{lipton2004ece}). However, complete, envy-free allocations with indivisible goods cannot always be guaranteed (e.g. with two agents and one good, the agent who did not get the good will always envy the other). Thus, we make use of two popular relaxations of EF as introduced by \cite{lipton2004ece}, \cite{budish2011ef1}, and \cite{caragiannis2016mnw}.

An allocation $\mathcal{A} = (A_1, \dots, A_n)$ is individually \emph{envy-free up to any good} (EFX) if, for every pair of agents $p_i, p_{i'} \in N$, and for all goods $g \in A_{i'}$,  
$v_{i}(A_i) \geq v_{i}(A_{i'}  \setminus \{ g \})$. 
Similarly, an allocation $\mathcal{A}$ is individually \emph{envy-free up to one good} (EF1) if, for every pair of agents $p_i,p_{i'}\in N$, there is \emph{some} good $g \in A_{i'}$ such that 
$v_{i}(A_i) \geq v_{i}(A_{i'} \setminus \{ g \})$.

\cite{chakraborty2020wef} recently introduced an extension of the EF notion to the weighted setting, known as \emph{weighted envy-freeness} (WEF). 
In this setting, agents represent groups where each group has a fixed weight. We use this notion to capture inter-group envy.
Similarly, we consider two relaxed notions of WEF. 
The definitions below rely on the assumption that the groups' valuations of a bundle are the same regardless of how goods are internally allocated according to $\mathcal{A}$; this is a valid assumption if we assume that valuation functions of agents within a group cannot differ. In Section \ref{sec:general}, we introduce an extension of the WEF notion to deal with the more general case. 
An allocation $\mathcal{A} = (A_1, \dots, A_n)$ is said to be \emph{weighted envy-free up to one good} (WEF1) if for every two groups $T_k,T_{k'} \in \mathcal{T}$, there exists some good $g \in B_{k'}$ such that $\frac{v_{T_k}(B_k)}{w_k} \ge \frac{v_{T_k}(B_{k'}\setminus \{ g \})}{w_{k'}}$. It is weighted envy-free up to \emph{any} good (WEFX) if this inequality holds for \emph{any} $g \in B_{k'}$.  

Note that envy-freeness and weighted envy-freeness are referred to as EF and WEF respectively in the literature, but we refer to them as \ef{} and \wef{} henceforth, to highlight that the former is an individual fairness concept, and the latter is a group fairness concept. We begin with an example to illustrate these fairness concepts.

\begin{example}
    Consider a setting in which we have two groups $T_1$ and $T_2$, consisting of one and two agents respectively, with $p_1 \in T_1$ and $p_2, p_3 \in T_2$.  Suppose that there are five goods $g_1, g_2, g_3, g_4, g_5$, for which all agents have equal valuation: $v(g_1) = v(g_2) = v(g_3) = v(g_4) = v(g_5) = c > 0$.

    For individual envy-freeness, suppose that an allocation $\mathcal{A}$ is such that $p_1$ has one good, and $p_2, p_3$ have two goods each. Then, $p_2, p_3$ clearly have no envy towards $p_1$, since $v(A_2) = v(A_3) = 2c > c = v(A_1)$. This inequality also indicates that $p_1$ has envy for each of $p_2$ and $p_3$. However, observe that if we were to remove \textit{one} good from each of $p_2$ and $p_3$'s bundle, then any envy $p_1$ has for each of the other agents would disappear. Hence, we say that $\mathcal{A}$ is \ef{1}.
    
    In the more general case of goods with different valuations, \ef{} allows choosing any single good to remove in the above manner, e.g., the most valued one.  In contrast, the stronger variant \ef{X} requires the envy-free condition to hold no matter which good was removed, e.g., the least-valued one.
    
    
    For group weighted envy-freeness, Considering the same allocation $\mathcal{A}$ as above, $T_1$'s weighted bundle value is $v(B_1) = \frac{v(A_1)}{w_1} = \frac{c}{1} = c$ and $T_2$'s weighted bundle value is $v(B_2) = \frac{v(A_2) +v(A_3)}{w_2} = \frac{4c}{2} = 2c$. Then, clearly, $T_2$ has no weighted envy towards $T_1$. However, the converse is not true. Observe that even if we remove any good (call it $g$) from $T_2$'s bundle, $\frac{v(B_2 \backslash \{g\})}{w_2} = \frac{3c}{2} > c = \frac{v(B_1)}{w_1}$, and there is still weighted envy by $T_1$ towards $T_2$. Hence, this allocation $\mathcal{A}$ is neither \wef{X} nor \wef{1}. 
    
    More generally, if the removal of some good from the bundle of $B_2$ gives $\frac{v(B_2 \backslash \{g \})}{w_2} \leq \frac{v(B_1)}{w_1}$, then we can say the allocation is \wef{1}. Similarly, if the same holds no matter which good is removed, then the allocation satisfies the more stringent  \wef{X} property.
\end{example}


\section{Approximate \ef{} and \wef{} Allocations}\label{sec:efwef}

In this section, we analyze the existence of approximate individual EF (\ef{}) and group WEF (\wef{}) allocations.  Additional notions of fairness are also discussed in Appendix D, but this section contains the results that we view as our most important contributions.

\subsection{All-Common Valuations} \label{sec:allcommon}

\ef{X} allocations are known to exist within the restricted setting of all-common valuations \cite{plaut2018efx} (i.e. when all agents have identical valuation functions). 
An interesting starting point is to study the existence of allocations that satisfy \ef{X} and approximate \wef{} simultaneously. 

A natural question is whether the concept of ``up to the least valued good'' can be extended to the weighted setting, and be achieved together with its individual counterpart.

We first address the first part of the question. As there have been no prior works looking into the the existence of \wef{X} allocations, we provide a basic result that extends a commonly known \ef{X} result to its weighted counterpart.

\begin{theorem}\label{thm-wefx}
    Under all-common, additive valuation functions, a \wef{X} allocation can be computed in polynomial time.
\end{theorem}

\begin{proof}
    Order the goods in decreasing order of preference, and sequentially allocate one to an agent with the least weighted bundle value (i.e. value of its bundle divide by its weight). Clearly, at any point in time, the agent that which we allocate a good to cannot be (weighted) envied pre-allocation. Any envy that forms towards that agent can only be as a result of the latest good allocated, and any such envy will disappear upon dropping this good. In addition, this good is also the least valued one in the agent's bundle by construction. The \wef{X} property follows, and it is easy to see that the algorithm runs in polynomial time.
\end{proof}

Then, for the second part, unfortunately, we show that it is generally not possible, with the following proposition.

\begin{proposition} \label{prop-wefxincompatiblewithef}
    \wef{X} is incompatible with approximate \ef{} notions (\ef{X} or \ef{1}) in general, even when all agents' valuation functions are identical.
\end{proposition}
\begin{proof}
    Consider an all-common valuation setting in which we have two groups $T_1$ and $T_2$, each with two agents $p_1, p_2 \in T_1$ and $p_3, p_4 \in T_2$. There are four goods, which all agents value equally: $v(g_1) = c$, $v(g_2) = v(g_3) = v(g_4) = 1$. Here, $c\gg 1$ is some very large constant value. Then, in order for an allocation to be \ef{1} or \ef{X}, each agent must receive exactly one good (otherwise, some agent will have two goods and some will have none, which clearly violates either of the approximate \ef{} notions). Without loss of generality, suppose $p_1 \in T_1$ gets $g_1$, and the rest of the agents receive a good of value $1$. In this case, $T_2$ has weighted envy towards $T_1$, even if the lowest valued good (of value 1) is removed from $T_1$'s bundle.
\end{proof}

Since \wef{X} is incompatible with both notions of approximate \ef{}, we focus on a slightly relaxed group fairness property: \wef{1}. Again, focusing on the case of all-common valuations, since \ef{X} is arguably the strongest relaxation of \ef{} \citep{caragiannis2019efxnash}, it is of interest to ask whether an \ef{X} allocation can guarantee \wef{1}. However, the following proposition shows that this is not the case.

\begin{proposition} \label{prop-efxnotimplywef1}
    An \ef{X} allocation is not necessarily \wef{1}, nor is a \wef{1} allocation necessarily \ef{X}, even when all agents' valuation functions are identical.
\end{proposition}
\begin{proof}
    For the first part, consider the all-common valuation setting in the case of two groups $T_1$ and $T_2$, with $p_1, p_2 \in T_1$ and $p_3, p_4 \in T_2$, and there are four goods, $v(g_1) = v(g_2) = c,  v(g_3) = v(g_4) = 1$. Here, $c\gg 1$ is some very large constant value. Then, any allocation where each agent gets exactly one good is \ef{X}. Consider the allocation where agent $p_i$ is allocated good $g_i$ for every $i \in \{1,\dots, 4 \}$ (i.e. $T_1$ receives the two valuable goods, and $T_2$ receives the two least valued goods). This allocation is not \wef{1}, as $T_2$ envies $T_1$ even when one of the goods is removed.
    
    For the second part, consider a setting where we have two groups $T_1, T_2$ ($w_1 = 2$ and $w_2 = 1$), and three goods $g_1, g_2, g_3$ with all-common valuations $v(g_1) = c, v(g_2) = v(g_3) = 1$. Again, $c\gg 1$ is some very large constant value. Note that the allocation $\mathcal{A}$ with $(B_1, B_2) = (\{g_1\}, \{g_2, g_3\})$ is \wef{1}. However, since $T_1$ has two agents but only one good, one of the agents in $T_1$ will receive nothing. In particular, since the agent in $T_2$ receives two goods, the empty-handed agent will envy them even after removing any good.
\end{proof}
Proposition \ref{prop-efxnotimplywef1} indicates that we can neither rely on existing \ef{X} algorithms (that do not take into consideration groups) such as in \cite{plaut2018efx} and \cite{aziz2020groupefchores}, nor can we directly make use of existing \wef{1} algorithms such as in \cite{chakraborty2020wef} to achieve both properties. 
We therefore propose the \emph{Sequential Maximin-Iterative Weighted Round Robin} (SM-IWRR) algorithm  (Algorithm \ref{algo-iwrrse}) that can, in the all-common valuation setting, provably produce an allocation that is both \ef{X} and \wef{1} in polynomial time.

\begin{algorithm}[tb]
\caption{Sequential Maximin-Iterative Weighted Round Robin (SM-IWRR)}
\label{algo-iwrrse}
\textbf{Input}: set of agents $N$, set of goods $G$, set of groups $\mathcal{T}$, valuation function $v$ \\
\textbf{Output}: allocation $\mathcal{A}$
\begin{algorithmic}[1]

\STATE Run the SM algorithm (see Algorithm \ref{algo-sm}) with inputs $N$, $G$ and $v$, and obtain output $\mathcal{A'} = (A_1', \dots, A_n')$
\STATE Let $A_\text{min}' = \argmin_{i:p_i \in N} v(A_i') $ be the lowest-valued bundle in $\mathcal{A'}$
\STATE Initialize set of \textit{representative goods}, $R = \{\}$
  \FOR{each $A_i' \in \mathcal{A}'$}
  \STATE Create a new good $r_i$, with $\hat{v}(r_i) = v(A'_i) - v(A'_\text{min})$\\
  \STATE $R \leftarrow R \cup \{r_i \}$
  \ENDFOR
  \STATE Run the IWRR algorithm (see Algorithm \ref{algo-iwrr}) with inputs $N$, $R$, $T$ and $\hat{v}$, and obtain output $\mathcal{A} = (A_1, \dots, A_n)$.
  \FOR{each $A_j \in \mathcal{A}$}
  \FOR{each $r_i \in R$}
  \IF{$r_i \in$ $A_j$}
  \STATE $A_j \leftarrow A_i'$
  \ENDIF
  \ENDFOR
  \ENDFOR
 \STATE \textbf{return} $\mathcal{A} = (A_1, \dots, A_n)$
\end{algorithmic}
\end{algorithm}

\begin{algorithm}[tb]
 \caption{Sequential Maximin (SM)}
 \label{algo-sm}
\textbf{Input}: set of agents $N$, set of goods $G$, valuation function $v$ \\
\textbf{Output}: allocation $\mathcal{A}$
\begin{algorithmic}[1]
  \STATE Initialize $A_i = \{\}$ for $i = 1, \dots, n$
 \WHILE{there are unassigned goods $G_{\text{unassigned}} \subseteq G$}
  \STATE Let $g = \argmax_{j:g_j \in G_{\text{unassigned}}} v(g_j)$ be the  highest-valued unassigned good
  \STATE Let $p_i \in N$ be the agent with the least-valued bundle $A_i$, where $A_i = \argmin_{j:p_j \in N} v(A_j)$
  \STATE $A_i \leftarrow A_i \cup \{g \}$
  \STATE $G_{\text{unassigned}} \leftarrow G_{\text{unassigned}} \setminus \{g\}$
 \ENDWHILE
 \STATE \textbf{return} $\mathcal{A} = (A_1, \dots, A_n)$
\end{algorithmic}
\end{algorithm}

\begin{algorithm}[tb]
 \caption{Iterative Weighted Round Robin (IWRR)}
 \label{algo-iwrr}
\textbf{Input}: set of agents $N$, set of goods $G$, set of groups $\mathcal{T}$, and set of valuation functions $\{ v_1, \dots, v_n \}$ \\
\textbf{Output}: allocation $\mathcal{A}$
\begin{algorithmic}[1]
\STATE Initialize $A_i = \{\}$ for $i = 1, \dots, n$ and $G_{\text{unassigned}} = G$
\WHILE{there are unassigned goods $G_{\text{unassigned}} \subseteq G$}
  \STATE Let $T_k \in \mathcal{T}$ be the group with the lowest weighted bundle size $\frac{|B_k|}{w_k}$
  \STATE Let $p_i \in T_k$ be the agent with the lowest bundle size $|A_i|$, with ties broken in favour of the one that has highest marginal utility from a good in $G_{\text{unassigned}}$ (i.e. tie-break in favour of the agent whose favourite unassigned good is of highest (subjective) valuation)
  \STATE Let $g \in G_{\text{unassigned}}$ be the good $p_i$ values the most
  \STATE $A_i \leftarrow A_i \cup \{g \}$
  \STATE $G_{\text{unassigned}} \leftarrow G_{\text{unassigned}} \setminus \{g\}$\;
\ENDWHILE
 \STATE \textbf{return} $\mathcal{A} = (A_1, \dots, A_n)$
\end{algorithmic}
\end{algorithm}

Intuitively, the SM-IWRR algorithm works by first assigning goods to agents via the SM algorithm\footnote{In fact, any \ef{X} allocation algorithm could be used in place of the SM algorithm.}, such that the resulting allocation is \ef{X} (as shown in Theorem \ref{thm-allcommon}). Then, since valuations are all-common, the algorithm takes each bundle and treats it as a single good, referred to as the \emph{representative good}. 
The value of each representative good is then reduced by the value of the least-valued representative good. These representative goods are then allocated to agents via the IWRR algorithm\footnote{Note that in this setting, since each agent gets exactly one good, in line 3, the algorithm is simply picking any agent without a good instead of one with the lowest bundle size. The algorithm is designed to be more general for use in later sections.} using these values. 
Each agent then receives the bundle corresponding to the representative good it was allocated. We first mention an important property about the algorithm.

\begin{proposition} \label{prop-seiwrrruntime}
    The SM-IWRR algorithm is guaranteed to terminate in $\mathcal{O}(m\log m)$ steps.
\end{proposition}
\begin{proof}
    For the SM algorithm, since valuations are all-common, finding the next favourite good of any agent can be made trivial via pre-processed sorting, which can be done in $O(m \log m)$ time. There are $\mathcal{O}(m)$ iterations of the \textbf{while} loop; in each iteration, determining the next agent to be given a good takes $\mathcal{O}(\log n)$ time. Thus, the SM algorithm runs in $\mathcal{O}(m (\log m + \log n))$ time. 
    
    For the IWRR algorithm, since our setting is such that each agent gets exactly one good, finding the next favourite good of any agent can be made trivial via pre-processed sorting, which can be done in $O(m \log m)$ time. There are $\mathcal{O}(n)$ iterations of the while loop; in each iteration, finding the next group takes $\mathcal{O}(\log \ell)$ time, and determining the next agent to be given a good is straightforward. Thus, the IWRR algorithm (one agent-one good variant) runs in $\mathcal{O}(m \log m + n \log \ell)$ time.
    
    In the SM-IWRR algorithm, since each of the \textbf{for} loops takes $\mathcal{O}(n)$ steps, assuming $m > n$ and combining these with the bounds above, the time complexity of the SM-IWRR algorithm is $\mathcal{O}(m\log m)$, and is guaranteed to terminate since the input sets are finite.
\end{proof}

The following theorem provides our first main result.
\begin{theorem} \label{thm-allcommon}
    Under all-common, additive valuation functions, the SM-IWRR algorithm returns an \ef{X} and \wef{1} allocation.
\end{theorem}
\begin{proof}
    We first prove that the SM-IWRR algorithm outputs an \ef{X} allocation. Consider the execution of the SM algorithm. Since at every round the algorithm selects an agent with the least valued bundle to be allocated a good $g \in G$, that agent cannot be envied prior to this allocation (because it has the least valued bundle from every other agent's perspective). 
    
    Hence, any envy that arises must be due to $g$, which is also the least valued good in that agent's bundle. This establishes the \ef{X} property. When we allocate representative goods via the IWRR (which preserves \ef{X}, since bundles are not modified), and map back into bundles, this amounts to a reallocation of bundles, so the resulting bundles that the agents receive are still \ef{X}.
    
    The \wef{1} property is established using the following result adapted from Theorem 3.3 of \cite{chakraborty2020wef}.
    
    \begin{lemma} \label{lemma-plwafp}
        Under all-common, additive valuation functions, the ``Pick the Least Weight-Adjusted Frequent Picker'' (PLWAFP) algorithm \citep{chakraborty2020wef} returns a \wef{1} allocation. 
    \end{lemma}
    
    The key point to note is that in IWRR, if we observe the group-level allocation, groups are allocated goods according to the least weight-adjusted frequency; thus, at the group level, IWRR is equivalent to PLWAFP, and in particular it ensures, by Lemma \ref{lemma-plwafp}, that the \wef{1} holds with respect to the $\hat{v}$ (i.e. the value function with respect to the representative goods).  However, more effort is needed to transfer this guarantee to the original valuation function $v$.
    

    After the \ef{X} allocation of goods to agents, consider the set of bundles $\{A_1,A_2,\dots,A_n\}$, where individual bundles are labelled such that $v(A_1 ) \geq v(A_2) \geq \dots \geq v(A_n)$. For all $p_i \in N$, define the \emph{representative good value} to be $\hat{v}(r_i)=v(A_i)-v(A_n)$, where $r_i$ is the representative good of $A_i$. Then, we make the following two claims:
    \begin{description}
        \item[Claim 1] For all $p_i \in N$, $\hat{v}(r_i)$ is upper bounded by the value of any one good in $A_i$.
        \item[Claim 2] For any two groups $T_k, T_{k'} \in \mathcal{T}$, let $B_k$ and $B_{k'}$ be the bundles of representative goods allocated to group $T_k$ and $T_{k'}$ respectively. If we have a \wef{1} allocation of representative goods to agents, then by replacing each representative good with its corresponding bundle, the allocation remains \wef{1}.
    \end{description}
    
    Claim 1 holds because the allocation $\mathcal{A}$ is \ef{X}, so for all $p_i \in N$ and any $g \in A_i$, $v(A_i\setminus \{g \}) \leq v(A_n)$. Then, as valuations are additive, $v(A_i)-v(\{g\}) \leq v(A_n)$, and hence $\hat{v}(r_i) = v(A_i) - v(A_n) \leq v(\{g\})$. Claim 2 implies that the replacement step (of each representative good by its bundle) in the SM-IWRR algorithm preserves the ``up to one good'' guarantee. We proceed to prove claim 2.

       Since we have a \wef{1} allocation on representative goods, for any two groups $T_k, T_{k'} \in \mathcal{T}$, there exists a representative good $r_\text{max} \in B_{k'}$, such that $\hat{v}(r_\text{max}) = \max_{i':p_{i'} \in T_{k'}} \hat{v}(r_{i'})$ and the following holds:
       \begin{equation}\label{wef1representativegoods}
           \frac{\sum_{i:p_i \in T_k} \hat{v}(r_i)}{w_k} \geq \frac{\sum_{i': p_{i'} \in T_{k'}} \hat{v}(r_{i'}) - \hat{v}(r_\text{max})}{w_{k'}}
       \end{equation}
       By the definition of a representative good, for all $i$ with $p_i \in T_k$, $\hat{v}(r_i) = v(A_i) - v(A_n)$, and recalling that $A_n$ is the least-valued bundle in $\mathcal{A}$, the left-hand side of (\ref{wef1representativegoods}) becomes $\frac{\sum_{i: p_i \in T_k} v(A_i)}{w_k} - v(A_n)$,
       and the right-hand side of (\ref{wef1representativegoods}) becomes $\frac{\sum_{i': p_{i'} \in T_{k'}} v(A_{i'}) - \hat{v}(r_\text{max})}{w_{k'}} - v(A_n)$.
       Combining these with (\ref{wef1representativegoods}), and $v(B_k) = \sum_{i: p_i \in T_k} v(A_i)$, we get
       \begin{equation*}
           \frac{v(B_k)}{w_k} = \frac{\sum_{i: p_i \in T_k} v(A_i)}{w_k} \geq \frac{\sum_{i': p_{i'} \in T_{k'}} v(A_{i'}) - \hat{v}(r_\text{max})}{w_{k'}}.
       \end{equation*}
       Then, since $v(B_{k'}) = \sum_{i': p_{i'} \in T_{k'}} v(A_{i'})$, it follows that 
       \begin{equation*}
           \frac{v(B_k)}{w_k} \geq \frac{v(B_{k'}) - \hat{v}(r_\text{max})}{w_{k'}} \geq \frac{v(B_{k'}) - v(\{g_\text{max}\})}{w_{k'}},
       \end{equation*}
      where the second inequality is a result of $\hat{v}(r_\text{max})$ being upper bounded by the value of one good in bundle $B_{k'}$ by claim 1 (specifically, the good of maximum value $g_\text{max} \in B_{k'}$).

    Claim 2 ensures that the \wef{1} property with respect to the \ef{X} bundles (or representative goods) transfers to the original goods, completing the proof of Theorem \ref{thm-allcommon}.
\end{proof}

\subsection{Group-Common Valuations} \label{sec:groupcommon}
Next, we consider the setting where agents in different groups may have different valuations, but agents \emph{within} any given group have the same valuations. We refer to this setting as one where agents have \emph{group-common valuations}. 
More formally, for each good $g \in G$, and any two agents $p_i, p_{i'} \in T_k$, $v_i(g) = v_{i'}(g)$. 

As the existence of \ef{X} allocations in this setting is still an open question \citep{plaut2018efx,caragiannis2019efxnash,chaudhury2020efxthree}, we explore \ef{1} and its compatibility with \wef{1}. We first give a proposition showing that \ef{1} allocations are not guaranteed to be \wef{1}. The proof of Proposition \ref{prop-ef1notimplywef1} is similar to that of Proposition \ref{prop-efxnotimplywef1}\footnote{In fact, Proposition \ref{prop-ef1notimplywef1} is a generalization of Proposition \ref{prop-efxnotimplywef1}.}, and is thus omitted.

\begin{proposition} \label{prop-ef1notimplywef1}
    An \ef{1} allocation is not necessarily \wef{1}, nor is a \wef{1} allocation necessarily \ef{1}, even when all agents' valuation functions are identical.
\end{proposition}
Much like the all-common valuation setting, we cannot rely on existing \ef{1} algorithms (that do not take into consideration groups) such as in \cite{lipton2004ece} to achieve both properties. 

We therefore propose that the \textit{Iterative Weighted Round Robin} (IWRR) algorithm (Algorithm \ref{algo-iwrr}), as introduced in the previous section, under group-common valuations, produces an allocation that is both \ef{1} and \wef{1} in polynomial time.

\begin{proposition} \label{prop-iwrrruntime}
    The IWRR algorithm is guaranteed to terminate in $\mathcal{O}(\ell m n)$ steps.
\end{proposition}
\begin{proof}
    There are $\mathcal{O}(m)$ iterations of the \textbf{while} loop. In each iteration, finding the next group takes $\mathcal{O}(\ell)$ time, determining the next agent to allocate takes $\mathcal{O}(n)$ time, while letting the agent pick its favourite good takes $\mathcal{O}(m)$ time (note that the logarithmic-time improvements as utilised in Proposition \ref{prop-seiwrrruntime} does not apply for non-identical valuations). Since the input sets are finite, the algorithm runs in $\mathcal{O}(\ell mn)$ time. 
\end{proof}
Next, we provide our main result for group-common valuations.

\begin{theorem} \label{thm-groupcommoniwrr}
    Under group-common, additive valuation, the IWRR algorithm returns an \ef{1} and \wef{1} allocation. 
\end{theorem}

\begin{proof}
   We first prove the \ef{1} property. Suppose we have a sequence of allocations of (agent, good) pairs for a single execution of the IWRR algorithm. We break this sequence into sub-sequences, which we will call \emph{rounds}, indexed by $r \in [1, K]$, with $K = \lceil \frac{m}{n} \rceil$. At the end of round $r$, each agent has $r$ goods in its bundle. While not all agents may receive a good in the final round, we can simply add dummy goods of value zero to complete the round. Let us denote by $g_{i, r}$ the good that agent $p_i$ picked in round $r$. We will show that each agent $p_i \in N$, can only possibly envy any other agent $p_{i'}$ up to the good $p_{i'}$ picked in round 1 (i.e. up to good $g_{{i'}, 1}$). 
   
   \textbf{Case 1: $p_{i'}$ selects after $p_i$ in the ordering.} For each round $r \in [1, K]$, since $p_i$ selects its most valued good, $v_i(g_{i, r}) \geq v_i(g_{{i'}, r})$. Therefore, $v_i(A_i) = \sum_{r=1}^K v_i(g_{i, r}) \geq \sum_{r=1}^K v_i(g_{{i'}, r}) = v_i(A_j)$.
   
   \textbf{Case 2: $p_{i'}$ selects before $p_i$ in the ordering.} For each round $r \in [1, K-1]$, since $p_i$ selects its most valued good, $v_i(g_{i, r}) \geq v_i(g_{{i'}, r+1})$. Therefore, $v_i(A_i) = \sum_{r=1}^{K} v_i(g_{i, r}) \geq \sum_{r=1}^{K-1} v_i(g_{i, r}) \geq \sum_{r=2}^{K} v_i(g_{{i'}, r}) = v_i(A_j \setminus \{g_{{i'}, 1}\})$. Since this holds for all such agents $p_i, p_{i'} \in N$, the allocation is \ef{1}.
   
   The proof of \wef{1} is the same as that of Theorem \ref{thm-allcommon}.
\end{proof}

\subsection{General Valuations} \label{sec:general}

We now proceed to study the existence of individual and group fair allocations under general additive valuations. Under this class of valuation functions, the distinction between the \wef{} notion defined in \cite{chakraborty2020wef} and our setting is more apparent. Agents within a group can have different valuations for each good, and so a key consideration in characterising \wef{} is, for any two groups $T_k, T_{k'} \in \mathcal{T}$, the valuation of a group $T_k$ for another group's $T_{k'}$ bundle. This was not a concern in the previous two (All-Common and Group-Common) settings, as the valuation of the bundle to a group was the same regardless of which agent within the group actually received the good. 

With this in mind, one could consider two methods of defining \wef{} -- an allocation-based approach (where valuations of a group for another group's bundle depends on some internal allocation procedure, as in \cite{benabbou2019groupfairness}), or a non-allocation based  approach (where valuations of a group for another group's bundle are quantified without reference to a specific internal allocation algorithm).

We will only consider a non-allocation based variant of \wef{}. In particular, we consider approximate \wef{} with respect to the bundle values of agents obtained in a randomized allocation, which we term \textit{ex-ante} \wef{1}. 
Similar (randomized) approaches for envy-freeness in the unweighted setting was considered in recent works \citep{freeman2020exanteexpost, aziz2020exanteexpost,babaioff2021bestof,kawase2022randomassignmentconstraints}.
\footnote{We note that the definition of approximate group fairness (``up to one good'' relaxation of removing one good per agent) used in \cite{conitzer2019groupfairness} and \cite{aziz2020groupefchores} would, in fact, be a weaker notion for consideration in the setting of pre-defined groups, and would arguably be too strong to satisfy should the relaxation be strengthened to the ``up one good'' per group variant. A randomized approach, as adopted in many recent works \citep{freeman2020exanteexpost, aziz2020exanteexpost,babaioff2021bestof,kawase2022randomassignmentconstraints}, thus appears to provide the best balance between simplicity, mathematical soundness, and intuition.} 
There has also been a long line of work in economics on ex-ante fairness and efficiency notions in this space \citep{bogomolnaia2001randomassignment,chen2002housingexperiment,nesterov2017sp}.
Intuitively, instead of assuming that items are allocated to all agents by some allocation procedure, we consider what the \emph{average} utility would be if we were to allocate each item to a \emph{uniformly random} agent. 



\begin{definition}[Ex-ante \wef{1}] \label{wef1expectation}
    An allocation $\mathcal{A} = (A_1, \dots, A_n)$ is ex-ante weighted envy-free up to one good (ex-ante \wef{1}) if, for every two groups $T_k, T_{k'} \in \mathcal{T}$, there exists some good $g \in B_{k'}$ such that $\frac{v_{T_k}(B_k)}{w_k} \geq \frac{\overline{v}_{T_k}( B_{k'} \setminus \{ g\})}{w_{k'}}$, where $\overline{v}_{T_k}(B_{k'}) = \frac{1}{w_k} \sum_{i:p_i \in T_k} \left( \sum_{g' \in B_{k'}} v_i(g') \right) $.
\end{definition}

Further relaxing this notion in a standard manner, we say that an allocation is ex-ante \wef{1} \emph{up to a factor of} $\frac{1}{\gamma}$ for some constant $\gamma$ when the condition in Definition \ref{wef1expectation} is replaced by $\frac{v_{T_k}( B_{k})}{w_k} \geq \frac{1}{\gamma} \cdot \frac{\overline{v}_{T_k}( B_{k'} \setminus \{ g\})}{ w_{k'}}$. We proceed to show that the IWRR algorithm can help us achieve an approximate notion of ex-ante \wef{1} under this setting. 

\begin{theorem} \label{thm-generaliwrr}
    Under general, additive valuation, the IWRR algorithm returns an \ef{1} allocation that is ex-ante \wef{1} up to a factor of $\frac{1}{3}$.
\end{theorem}
\begin{proof}
    We have shown in Theorem \ref{thm-groupcommoniwrr} that the allocation returned by the IWRR algorithm is \ef{1}. We will thus focus on proving the approximate ex-ante \wef{} property.
    
    Consider the sequence of allocations to any two groups $T_k, T_{k'} \in \mathcal{T}$ in a single execution of the IWRR algorithm. Break this sequence into sub-sequences called \emph{rounds}, indexed by $r \in [1, K]$, with $K = \lceil \frac{|B_k| + |B_{k'}|}{w_k + w_{k'}} \rceil$. At the end of round $r$, each agent has $r$ goods in its bundle. While not all agents may receive a good in the final round, we can simply add dummy goods of value zero to complete the round.
    
    Each round is made up of $(w_k + w_{k'})$ \emph{iterations}, whereby one good is selected by some agent at each iteration. Assume that ties are broken in favour of $T_{k'}$ (the other case is similar). Then, the first iteration of every round is an allocation of a good to an agent in $T_{k'}$. In particular, this agent from $T_{k'}$ that gets to make a selection on the very first iteration of the first round could have picked a good that is of arbitrarily high value to every agent in $T_k$ -- this is the good that we will drop, as part of the ``up to one good'' relaxation.
    
    \textbf{Case 1: $w_k < w_{k'}$.} Here, we define a \textit{shifted round} $r$, that consists of all iterations (except the first) of the original round $r$, and the first iteration of the next round $r+1$. Note that if a round $r+1$ doesn't exist, we can simply add dummy goods of value zero such that the setting is well-defined. In the first round, we mentioned above that the first good is dropped (let this be $g_1$); every other good is accounted for in some shifted round. Figure \ref{fig:3.3fig1} illustrates a single shifted round $r$. We will argue the satisfiability of ex-ante \wef{1} up to a factor of $\frac{1}{3}$ in this one shifted round; the analysis then extends to all shifted rounds similarly.
    
    
    \begin{figure}[h]
      \centering
      \includegraphics[width=300pt]{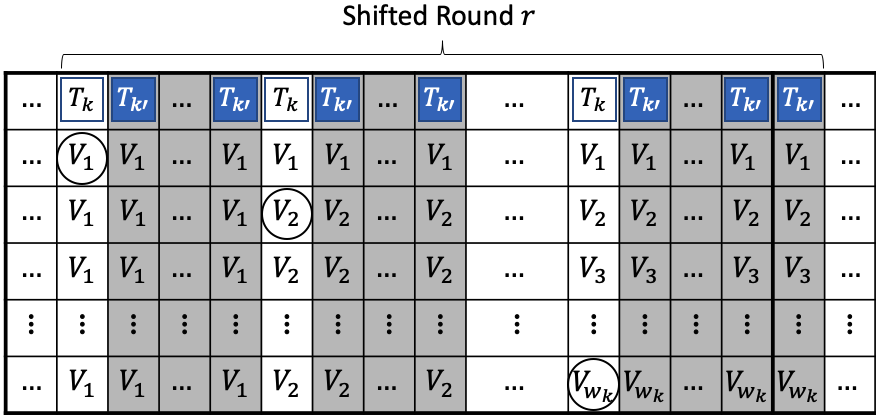}
      \caption{$T_k$'s agent valuations for goods in $B_k \cup B_{k'}$ for a shifted round $r$}
      \label{fig:3.3fig1}
    \end{figure}
    
    Let each entry $(i, j)$ in the matrix illustrated in Figure \ref{fig:3.3fig1} be the valuation that an agent $p_i \in T_k$ (row) has for good $g_j \in B_{k'}$ (column). Let the unshaded columns correspond to iterations whereby agents from $T_k$ make a selection, whereas shaded columns represent the iterations whereby an agent from $T_{k'}$ makes a selection. Without loss of generality, we can label goods and agents in a such way whereby the circled cells in unshaded columns represent the value of the good that was selected by the corresponding agent from $T_k$.
    
    Then, we have that in every shifted round, we are considering a sequence of selections of the following form: $T_k$ selects (unshaded column), $T_{k'}$ selects multiple consecutive times (shaded columns), $T_k$ selects (unshaded column), $T_{k'}$ selects multiple consecutive times (shaded columns), etc. The structure applies due to the following lemma, whose proof is given in Appendix B.

    \begin{lemma} \label{lemma-generalval}
        For any two groups $T_k, T_{k'} \in \mathcal{T}$, suppose $w_{k'} \geq w_k$. Then, when running the IWRR algorithm, $T_k$ cannot make more than one consecutive selection, whereas $T_{k'}$ can only make either $\floor{w_{k'}/w_k}$ or $\ceiling{w_{k'}/w_k}$ consecutive selections, at any point in time (ignoring selections by groups other than $T_k$ and $T_{k'}$). 
    \end{lemma}
    
    The values in uncircled cells indicate the maximum valuation that the agent in that row can have for the good in that column (as a result of how the algorithm works -- at every iteration, the agent from $T_k$ selects a good that gives the group maximum marginal valuation). The general observation we can make is that for every circled cell with value $V_i$, all the cells below it on the same column, and to its right -- on the same row, or below it -- cannot exceed $V_i$ in value. This is because we labelled and arranged agents such that those who make a selection earlier is on a higher row, and goods are picked in order from left to right.
    
    For each shifted round $r \in [1, K]$, let the set of circled cells be $S_{C}^r$ and the set of shaded cells be $S_{B}^r$. In addition, define their respective sum of cell values as $u(S_{C}^r)$ and $u(S_B^r)$. Then, by applying this concept to all shifted rounds, we have from Definition \ref{wef1expectation} that
    \begin{equation*}
        v_{T_k}(B_k) = \sum_{r = 1}^K u(S_{C}^r), \text{ and } \ \overline{v}_{T_k}(B_{k'} \setminus \{g_1 \}) = \frac{1}{w_k} \sum_{r = 1}^K u(S_{B}^r).
    \end{equation*}
    Then, in order to show the ex-ante \wef{1} up to a factor of $\frac{1}{3}$ property,
    it is equivalent to show that 
    \begin{equation} \label{eqn-case1-boundsum}
        3w_{k'} \sum_{r = 1}^K u(S_{C}^r) \geq \sum_{r = 1}^K u(S_{B}^r).
    \end{equation}

    \noindent In other words, if we can show that at every shifted round, the sum of shaded cells is upper bounded by ($3w_{k'} \times$ sum of circled cells), the desired property follows. In the following, in every shifted round $r$, for every circled cell $V_i^r \ (i \in [1, w_k])$, define $V_{i, ROW}^r$ and $V_{i, BLK}^r$ as follows:
    \begin{enumerate}
         \item $V_{i, ROW}^r$ = sum of shaded cells in the same row and to the right of circled cell $V_i^r$;
         \item $V_{i, BLK}^r$ = sum of shaded cells in the shaded columns sandwiched between the unshaded columns containing the circled cells $V_i^r$ and $V_{i+1}^r$, and starting from the row immediately below circled cell $V_i^r$.
     \end{enumerate}
    An example is illustrated in Figure \ref{fig:3.3fig2}.
    \begin{figure}
      \centering
      \includegraphics[width=300pt]{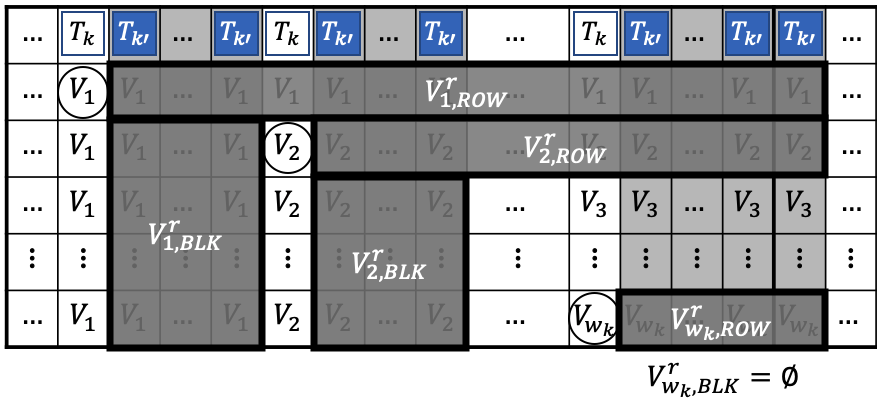}
      \caption{$T_k$'s agent valuations for goods in $B_k \cup B_{k'}$ for a shifted round $r$, with additional quantities indicated using shaded regions}
      \label{fig:3.3fig2}
    \end{figure}
    Also, by the definition of $u$, we have that $u(S_{C}^r) = \sum_{i=1}^{w_k} V_i^r$ and $u(S_B^r) = \sum_{i=1}^{w_k} \left(V_{i, ROW}^r +  V_{i, BLK}^r \right)$.
    Moreover, since $V_{i, BLK}^r$ has a maximum of $\ceiling{w_{k'}/w_k}$ columns (by Lemma \ref{lemma-generalval}) and a maximum of $w_k$ rows (because there are $w_k$ agents in group $T_k$), we have $V_{i, BLK}^r \leq V_i^r \times \ceiling{w_{k'}/w_k} \times w_k \leq 2 w_{k'} \times V_i^r$. The last inequality follows because of the assumption that $w_k \leq w_{k'}$ in this case.

    Then, coupled with the fact that for every $i \in [1, w_k]$, $V_{i, ROW}^r \leq w_{k'} \times V_i^r$ (because there is a maximum of $w_{k'}$ shaded columns in any single shifted round), we obtain 
    \begin{equation} \label{eqn-rowcol3wk'}
        V_{i, ROW}^r +  V_{i, BLK}^r \leq 3w_{k'} \times V_i^r
    \end{equation}

    By summing (\ref{eqn-rowcol3wk'}) on both sides over all $i \in [1, w_k]$ and shifted rounds $r \in [1, K]$,
    \begin{equation*}
        \sum_{r=1}^K \sum_{i = 1}^{w_k} \left( V_{i, ROW}^r +  V_{i, BLK}^r \right) \leq \sum_{r=1}^K \left( 3w_{k'} \times V_i^r \right)
    \end{equation*}
    and since $u(S_{C}^r) = \sum_{i=1}^{w_k} V_i^r$ and $u(S_B^r) = \sum_{i=1}^{w_k} \left(V_{i, ROW}^r +  V_{i, BLK}^r \right)$, we get (\ref{eqn-case1-boundsum}) as desired.
    
    \textbf{Case 2: $w_k \geq w_{k'}$:}
     This case is proven in a similar vein as the previous case. The difference is that instead of comparing the values of shaded cells with a single circled cell, we compare with the sum of cell values of a set of circled cells, and the bound follows. The details are given in Appendix C.
\end{proof}
The existence of a better approximation bound that can be obtained with comparable efficiency remains an open question, and largely depends on the allocation algorithm.


\section{Conclusions and Future Work} \label{sec:conclusion}
In this work, we show that individual fairness may come at the cost of group fairness. Group fairness is a great way to ensure diversity in outcomes \citep{benabbou2019groupfairness}. Our work attempts to reconcile diversity with individual demands.
We study the existence of allocations that satisfy individual and group (weighted) envy-freeness simultaneously, and show that when agents' additive valuations are identical or at least common within groups, existing approximations of envy-freeness at both individual and group levels are compatible and achievable concurrently. 
In the case of general, additive valuations, in mandating \ef{1}, the IWRR algorithm achieves ex-ante \wef{1} up to a factor of $\frac{1}{3}$. Our results thus shed light on the difficulty in achieving existing notions of individual and group fairness concurrently in more complex settings. 

In Appendix D, we additionally include a discussion on two new notions of fairness -- PEF and Group Stability -- that exploit the group structure inherent in numerous problem domains.  We show that both the SM-IWRR and IWRR algorithms achieve relaxed variants of these properties in addition to their individual and group fairness guarantees.

Possible future research includes delving into allocation-based definitions of \wef{1} to explore the existence of approximately fair allocations under that setting. It would also be interesting to consider non-sequential allocation mechanisms \citep{barman2018ef1po}, and to understand whether better bounds exist for the case of general additive valuations. In addition to the possibility of extending the analysis to the setting with chores \citep{aziz2019chores,aziz2020groupefchores}, or considering rules such as maximum weighted Nash welfare \citep{suksompong2022mwnw}, incorporating efficiency notions such as Pareto-optimality or exploring alternative valuation classes \citep{benabbou2020submodular} are potential avenues for future work. 



\bibliographystyle{ACM-Reference-Format} 
\bibliography{abb,sample}

\newpage
\appendix

\section{Proof of Lemma \ref{lemma-generalval}: Consecutive Selections in the IWRR algorithm}
Recall that $w_{k'} \geq w_k$ by assumption. We begin by proving the first part of the lemma, that is, $T_k$ cannot make more than one consecutive selection at any point in time. 
\subsection*{Proof of first part}
We want to show that a case of $$..., T_{k'} \text{ picks} , T_k \text{ picks}, T_k \text{ picks}, ...$$ cannot happen. Let $|B_k|, |B_{k'}|$ be the corresponding bundle sizes of $T_k$ and $T_{k'}$ respectively before $T_{k'}$ makes such a selection as above.
 
\textbf{Case 1: Ties broken in favour of $T_{k'}$.} It must be that $\frac{|B_{k'}|}{w_{k'}} \leq \frac{|B_k|}{w_k}$. Then after $T_{k'}$ makes a selection, since we assume $T_k$ chooses next, $\frac{|B_{k'}| + 1}{w_{k'}} >  \frac{|B_k|}{w_k}$. Suppose, for a contradiction, that $T_k$ makes more than one consecutive selection. Then it must be that $\frac{|B_{k'}| + 1}{w_{k'}} >  \frac{|B_k| + 1}{w_k}$, so that  $T_k$ can continue to make the second selection. However, this implies $\frac{|B_{k'}|}{w_{k'}} + \frac{1}{w_k} \geq \frac{|B_{k'}| + 1}{w_{k'}} >  \frac{|B_k| + 1}{w_k}$, and cancelling $\frac{1}{w_k}$ on both ends gives us  $\frac{|B_{k'}|}{w_{k'}} > \frac{|B_k|}{w_k}$, which is a contradiction.

\textbf{Case 2: Ties broken in favour of $T_k$.} It must be that $\frac{|B_{k'}|}{w_{k'}} < \frac{|B_k|}{w_k}$. Then after $T_{k'}$ makes a selection, since we assume $T_k$ chooses next, $\frac{|B_{k'}|}{w_{k'}} \geq  \frac{|B_k| + 1}{w_k}$. Suppose, for a contradiction, that $T_k$ makes more than one consecutive selection. Then it must be that $\frac{|B_{k'}| + 1}{w_{k'}} \geq  \frac{|B_k| + 1}{w_k}$, so that $T_{k'}$ can continue to make the second selection. However, this implies $\frac{|B_{k'}|}{w_{k'}} + \frac{1}{w_k} \geq \frac{|B_{k'}| + 1}{w_{k'}} \geq  \frac{|B_k| + 1}{w_k}$, and cancelling $\frac{1}{w_k}$ on both ends gives us $\frac{|B_{k'}|}{w_{k'}} \geq \frac{|B_k|}{w_k}$, which is a contradiction.
\subsection*{Proof of second part}
We now prove the second part of the lemma, that is, $T_{k'}$ can only make either $\floor{\frac{w_{k'}}{w_k}}$ or $\ceiling{\frac{w_{k'}}{w_k}}$ consecutive selections at any point in time.

The above has shown that $T_k$ can only make one selection at a time. Hence, we want to show that a case of \small
$$..., T_{k'}  \text{ picks}, T_k \text{ picks}, T_{k'} \text{ picks $H$ consecutive times}, T_k \text{ picks}, ..., $$ \normalsize
where $H < \floor{\frac{w_{k'}}{w_k}}$ or $H > \ceiling{\frac{w_{k'}}{w_k}}$, cannot happen. Let $|B_k|, |B_{k'}|$ be the bundle size of $T_k$ and $T_{k'}$ before $T_{k'}$ makes the first such selection in the sequence above.

\textbf{Case 1: Ties broken in favour of $T_{k'}$.} It must be that
\begin{equation}\label{firstcontradictionlemma}
    \frac{|B_{k'}|}{w_{k'}} \leq \frac{|B_k|}{w_k}.
\end{equation}
Then, after $T_{k'}$ makes a selection, since we assume $T_k$ chooses next,
\begin{equation}\label{secondcontradictionlemma}
    \frac{|B_{k'}|+1}{w_{k'}} > \frac{|B_{k}|}{w_{k}}.
\end{equation} 
Thereafter, $T_k$ gets to make a selection and it would be $T_{k'}$'s turn again, thus $\frac{|B_{k'}|+1}{w_{k'}} \leq \frac{|B_{k}|+1}{w_{k}}$. Suppose that $H < \floor{\frac{w_k}{w_{k'}}}$. That means that after $T_{k'}$ makes $\floor{\frac{w_{k'}}{w_k}}-1$ consecutive selections, $\frac{|B_{k'}| + 1 +\floor{\frac{w_{k'}}{w_k}}-1}{w_{k'}} > \frac{|B_k|+1}{w_k}$, so that $T_{k'}$ cannot make the $\left( \floor{\frac{w_{k'}}{w_k}}\right)^\text{th}$ selection, because it's $T_k$'s turn. However, this implies 
\begin{align*}
    \frac{|B_{k'}|}{w_{k'}} > \frac{|B_{k}|+1}{w_{k}} - \frac{\floor{\frac{w_{k'}}{w_k}}}{w_{k'}} & \geq \frac{|B_{k}|+1}{w_{k}} - \frac{\left(\frac{w_{k'}}{w_k}\right)}{w_{k'}} \\
    & = \frac{|B_{k}|+1}{w_{k}} - \frac{1}{w_k} = \frac{|B_k|}{w_k},
\end{align*}
which contradicts (\ref{firstcontradictionlemma}). Now suppose that $H > \ceiling{\frac{w_k}{w_{k'}}}$. That means that $T_{k'}$ makes $\ceiling{\frac{w_{k'}}{w_k}}$ consecutive selections and yet $\frac{|B_{k'}| + 1 + \ceiling{\frac{w_{k'}}{w_k}}}{w_{k'}} \leq \frac{|B_k|+1}{w_k}$, so that $T_{k'}$ can continue to make the $\left( \ceiling{\frac{w_{k'}}{w_k}} + 1 \right)^\text{th}$ selection. 
However, this implies 
\begin{align*}
    \frac{|B_{k'}|+1}{w_{k'}} \leq \frac{|B_{k}|+1}{w_{k}} - \frac{\ceiling{\frac{w_{k'}}{w_k}}}{w_{k'}} & \leq \frac{|B_{k}|+1}{w_{k}} - \frac{\left(\frac{w_{k'}}{w_k}\right)}{w_{k'}} \\
    & = \frac{|B_{k}|+1}{w_{k}} - \frac{1}{w_k} = \frac{|B_k|}{w_k},
\end{align*}
which contradicts (\ref{secondcontradictionlemma}).

\textbf{Case 2: Ties broken in favour of $T_k$.} 
It must be that
\begin{equation}\label{firstcontradictionlemma2}
    \frac{|B_{k'}|}{w_{k'}} < \frac{|B_k|}{w_k}.
\end{equation}
Then, after $T_{k'}$ makes a selection, since we assume $T_k$ chooses next,
\begin{equation}\label{secondcontradictionlemma2}
    \frac{|B_{k'}|+1}{w_{k'}} \geq \frac{|B_{k}|}{w_{k}}.
\end{equation} 
Thereafter, $T_k$ gets to make a selection and it would be $T_{k'}$'s turn again, thus $\frac{|B_{k'}|+1}{w_{k'}} < \frac{|B_{k}|+1}{w_{k}}$. Suppose that $H < \floor{\frac{w_k}{w_{k'}}}$. That means that after $T_{k'}$ makes $\floor{\frac{w_{k'}}{w_k}}-1$ consecutive selections, $\frac{|B_{k'}| + 1 +\floor{\frac{w_{k'}}{w_k}}-1}{w_{k'}} \geq \frac{|B_k|+1}{w_k}$, so that $T_{k'}$ cannot make the $\left( \floor{\frac{w_{k'}}{w_k}}\right)^\text{th}$ selection, because it's $T_k$'s turn. However, this implies 
\begin{align*}
    \frac{|B_{k'}|}{w_{k'}} \geq \frac{|B_{k}|+1}{w_{k}} - \frac{\floor{\frac{w_{k'}}{w_k}}}{w_{k'}} & \geq \frac{|B_{k}|+1}{w_{k}} - \frac{\left(\frac{w_{k'}}{w_k}\right)}{w_{k'}} \\
    & = \frac{|B_{k}|+1}{w_{k}} - \frac{1}{w_k} = \frac{|B_k|}{w_k},
\end{align*}
which contradicts (\ref{firstcontradictionlemma2}). Now suppose that $H > \ceiling{\frac{w_k}{w_{k'}}}$. That means that $T_{k'}$ makes $\ceiling{\frac{w_{k'}}{w_k}}$ consecutive selections and yet $\frac{|B_{k'}| + 1 + \ceiling{\frac{w_{k'}}{w_k}}}{w_{k'}} < \frac{|B_k|+1}{w_k}$, so that $T_{k'}$ can continue to make the $\left( \ceiling{\frac{w_{k'}}{w_k}} + 1 \right)^\text{th}$ selection. 
However, this implies 
$$\frac{|B_{k'}|+1}{w_{k'}} < \frac{|B_{k}|+1}{w_{k}} - \frac{\ceiling{\frac{w_{k'}}{w_k}}}{w_{k'}} \leq \frac{|B_{k}|+1}{w_{k}} - \frac{\left(\frac{w_{k'}}{w_k}\right)}{w_{k'}} = \frac{|B_k|}{w_k},$$ which contradicts (\ref{secondcontradictionlemma2}).


\section{Proof of Theorem \ref{thm-generaliwrr} (Case 2): \ef{1} and ex-ante \wef{1} property of the IWRR algorithm}

Recall that we focused our analysis on any two groups $T_k, T_{k'} \in \mathcal{T}$ to show the \wef{1} property of the IWRR algorithm. Case 1 handled the scenario where $w_k < w_{k'}$, we will handle the scenario where $w_k \geq w_{k'}$ in this case.

Similar to case 1, we define a \textit{shifted round} $r$, that consists of all iterations (except the first) of the original round $r$, and the first iteration of the next round $r + 1$. Note that if a round $r+1$ doesn't exist, we can simply add dummy goods of value zero such that the setting is well-defined. In the first round, we mentioned that the first good is dropped (let this be $g_1$); every other good is accounted for in some shifted round.

We make use of Figure \ref{fig:3.3fig3} to aid in our argument -- it illustrates a single shifted round $r$. We will argue the satisfiability of \wef{1} in expectation up to a factor of $\frac{1}{3}$ in this one shifted round; the analysis then extends to all shifted rounds similarly.

\begin{figure}[h]
  \centering
  \includegraphics[width=300pt]{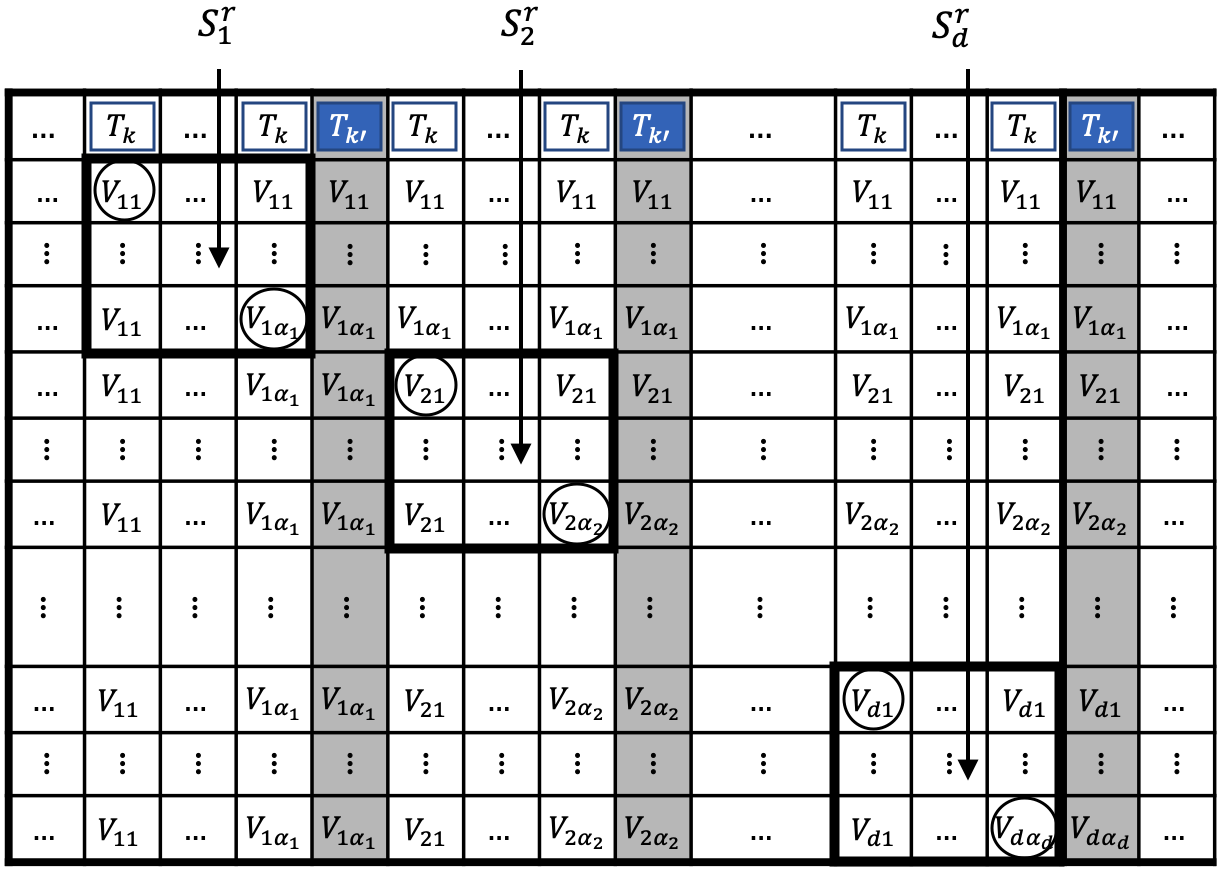}
  \caption{$T_k$'s agent valuations for goods in $B_k \cup B_{k'}$ for a shifted round $r$}
  \label{fig:3.3fig3}
\end{figure}

Let each entry $(i, j)$ in the matrix illustrated in Figure \ref{fig:3.3fig3} above be the valuation that an agent $p_i \in T_k$ (row) has for good $g_j \in B_{k'}$ (column). Let the unshaded columns correspond to iterations whereby agents from $T_k$ make a selection, whereas shaded columns represent the iterations whereby an agent from $T_{k'}$ makes a selection. Without loss of generality, we can label goods and agents in a such way whereby the circled cells belonging to unshaded columns represent the value of the good that was selected by the corresponding agent from $T_k$.
    
Then, we have that in every shifted round, we are considering a sequence of selections of the following form: $T_k$ selects multiple consecutive times (unshaded columns), $T_{k'}$ selects (shaded column), $T_k$ selects multiple consecutive times (unshaded columns), $T_{k'}$ selects (shaded column), etc. This structure applies due to Lemma \ref{lemma-generalval}.

The values in uncircled cells indicate the maximum valuation that the agent in that row can have for the good in that column (as a result of how the algorithm works -- at every iteration, the agent from $T_k$ selects a good that gives the group maximum marginal valuation). The general observation we can make is that for every circled cell with value $V$, all the cells below it on the same column, and to its right -- on the same row, or below it -- cannot exceed $V$ in value. This is because we labelled and arranged agents such that those who make a selection earlier is on a higher row, and goods are picked in order from left to right.
    
For each shifted round $r \in [1, K]$, let the set of circled cells be $S_{C}^r$ and the set of shaded cells be $S_{B}^r$. In addition, define their respective sum of cell values as $u(S_{C}^r)$ and $u(S_B^r)$. Then, by applying this concept to all shifted rounds, we have from Definition \ref{wef1expectation} that
\begin{equation*}
    v_{T_k}(B_k) = \sum_{r = 1}^K u(S_{C}^r), \text{ and } \ \overline{v}_{T_k}(B_{k'} \setminus \{g_1 \}) = \frac{1}{w_k} \sum_{r = 1}^K u(S_{B}^r).
\end{equation*}
Then, in order to show 
\begin{equation}\label{eqn-case2-approxwef1restate}
    \frac{v_{T_k}(B_k)}{w_k} \geq \frac{\overline{v}_{T_k}(B_{k'} \setminus \{g_1\})}{3w_{k'}},
\end{equation}
it is equivalent to show that 
\begin{equation} \label{eqn-case2-boundsum}
    3w_{k'} \sum_{r = 1}^K u(S_{C}^r) \geq \sum_{r = 1}^K u(S_{B}^r).
\end{equation}

\noindent In other words, if we can show that at every shifted round, the sum of shaded cells is upper bounded by ($3w_{k'} \times$ sum of circled cells), the property in (\ref{eqn-case2-approxwef1restate}) follows. 
    
The difference between this case and the previous is that now, there can be multiple consecutive unshaded columns, and hence circled cells are no longer isolated -- so our analysis will not simply be considering for every single circled cell, but for every set of circled cells in consecutive unshaded columns. For instance, as illustrated in Figure \ref{fig:3.3fig3}, in shifted round $r$, the first set is $S_1^r = \{ V_{11}, ..., V_{1\alpha_1} \}$, second set is $S_2^r = \{V_{21}, ..., V_{2\alpha_2} \}$, etc. Let there be $d$ sets, and let $u(S_i^r)$ be the sum of values in a set (i.e. sum of circled cells' values in a set $S_i^r$). For all $i \in [1, d]$, $\alpha_i$ denotes the number of circled cells in set $S^r_i$.
    
In the following, in every shifted round $r$, for every set of circled cells $S_i^r \ (i \in [1, d])$, define $S_{i, BLK}^r$ and $S_{i, COL}^r$ as follows:
\begin{enumerate}
    \item $S_{i, BLK}^r$ = set of shaded cells in the same row and to the right of every circled cell in $S_i^r$ (in shifted round $r$); 
         \item $S_{i, COL}^r$ = set of shaded cells in the first shaded column to the right of $S_i^r$, and starting from the row immediately below circled cell $V_{i\alpha_i}$.
\end{enumerate}
An example is illustrated in Figure \ref{fig:3.3fig4}.
\begin{figure}[h]
    \centering
    \includegraphics[width=300pt]{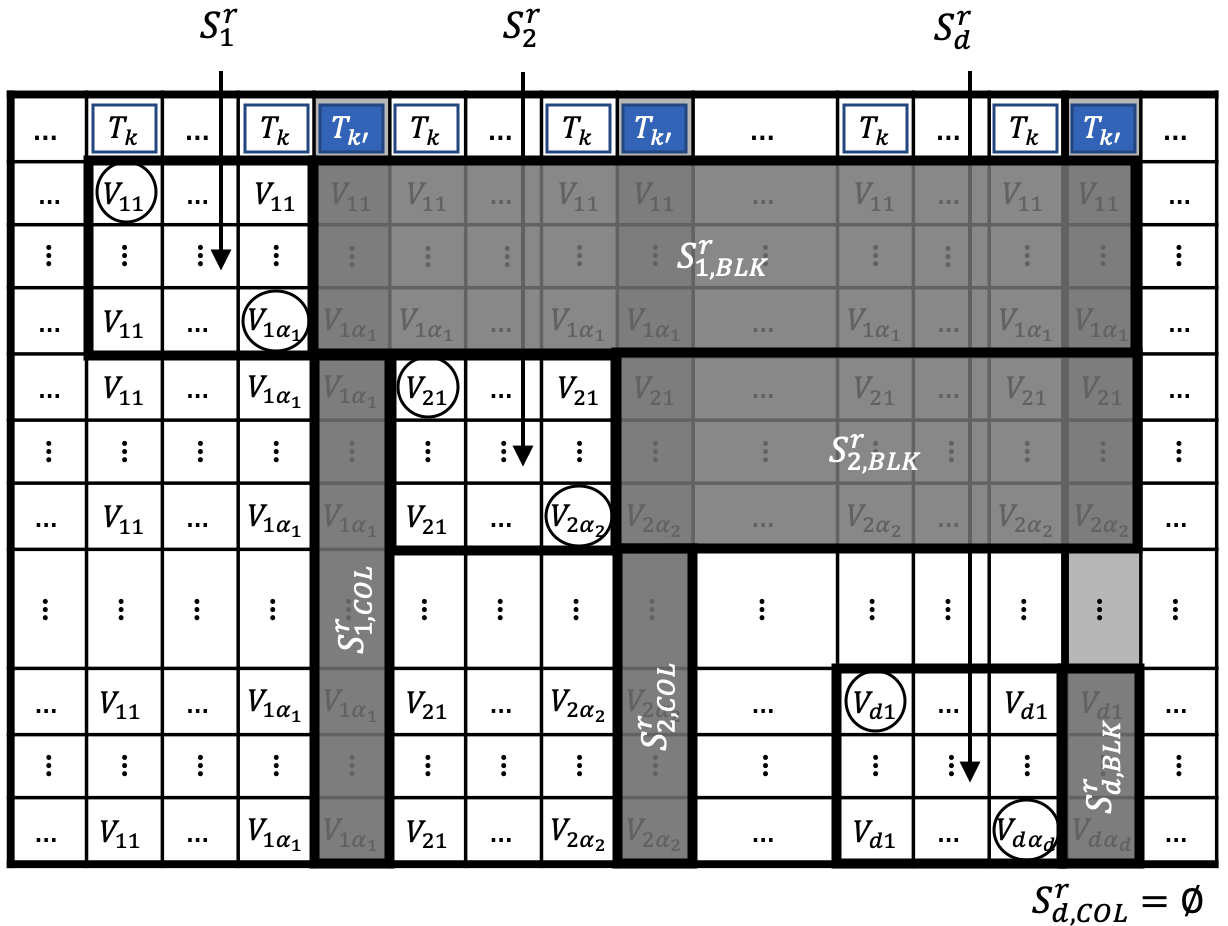}
    \caption{$T_k$'s agent valuations for goods in $B_k \cup B_{k'}$ for a shifted round $r$, with additional quantities indicated using shaded regions}
    \label{fig:3.3fig4}
\end{figure}
    
By the definition of $u$, we have the following:
\begin{enumerate}
    \item $u(S_{C}^r) = \sum_{i=1}^d u(S_i^r)$;
    \item $u(S_B^r) = \sum_{i=1}^d \left(u(S_{i, BLK}^r) +  u(S_{i, COL}^r) \right)$.
\end{enumerate}
Moreover, for any shifted round $r$, since $S_i^r = \{ V_{i1}, ..., V_{i\alpha_i} \}$, there are $\alpha_i$ circled cells forming $S_i$. Then, 
\begin{equation}\label{eqn-setblock-1-case2}
    u(S_{i, BLK}^r) \leq w_{k'} \times \left( V_{i1} + ... + V_{i\alpha_1} \right) = w_{k'} \times u(S_i^r),
\end{equation}
because there are a maximum of $w_{k'}$ shaded columns in a single shifted round. Next, for each $i \in [1, d]$, $\alpha_i \geq \floor{\frac{w_k}{w_{k'}}}$ (by Lemma \ref{lemma-generalval}). Thus, 
    \begin{equation} \label{eqn-setcol-1-case2}
    \begin{split}
        u(S_{i, COL}^r) & \leq V_{i \alpha_i} (w_k - \alpha_i)\\ 
        & \leq V_{i \alpha_i} \left( w_k -  \floor{\frac{w_k}{w_{k'}}} \right)  
        \\
        & \leq V_{i \alpha_i} \left( \frac{w_k (w_{k'} -1)}{w_{k'}} + 1 \right)  
        \\
        & \leq V_{i \alpha_i} \left( \left( \floor{\frac{w_k}{w_{k'}}} + 1 \right)  \left( w_{k'} - 1 \right) + 1 \right) 
        \\
        & \leq V_{i \alpha_i} \left( (\alpha_i + 1)(w_{k'} - 1) + 1 \right) 
        \\
        & = V_{i \alpha_i} \left( \alpha_i w_{k'} + w_{k'} - \alpha_i \right) \\
        & \leq V_{i \alpha_i} \left( 2 \alpha_i w_{k'} \right) 
        \\
        & \leq 2 w_{k'} \times u(S_i^r),
    \end{split}
    \end{equation}
    where the last inequality is derived from the fact that $V_{i1} \geq V_{i2} \geq ... \geq V_{i\alpha_i}$ and $u(S_i^r) = \sum_{j=1}^{\alpha_i} V_{ij}$, implying $\alpha_i V_{i\alpha_i} \leq u(S_i^r)$.
    
    Then, combining (\ref{eqn-setblock-1-case2}) and (\ref{eqn-setcol-1-case2}), we obtain 
    \begin{equation} \label{eqn-rowcol3wk'-case2}
        u(S_{i, BLK}^r) +  u(S_{i, COL}^r) \leq 3w_{k'} \times u(S_i^r)
    \end{equation}

By summing (\ref{eqn-rowcol3wk'-case2}) on both sides over all $i \in [1, w_k]$ and shifted rounds $r \in [1, K]$,
\begin{equation}
    \sum_{r=1}^K \sum_{i = 1}^{w_k} \left( u(S_{i, BLK}^r) +  u(S_{i, COL}^r) \right) \leq \sum_{r=1}^K \left( 3w_{k'} \times u(S_i^r) \right)
\end{equation}
    and since $u(S_{C}^r) = \sum_{i=1}^d u(S_i^r)$ and $u(S_B^r) = \sum_{i=1}^d \left(u(S_{i, BLK}^r) +  u(S_{i, COL}^r) \right)$, it follows that 
    \begin{equation*}
        \sum_{r=1}^K u(S_B^r) \leq 3w_{k'} \sum_{r=1}^K u(S_{C}^r),
    \end{equation*}
    which gives (\ref{eqn-case2-boundsum}) as desired.

\section{Discussion on Additional Notions of Fairness} \label{sec:further}
Traditional notions of individual fairness have recently seen their group counterparts introduced \cite{aziz2019wmms,chakraborty2020wef}. However, when we look at allocating to individuals in groups, new opportunities emerge for us to characterise fairness notions specific to this setting. In addition to our studies on attaining individual and group fairness simultaneously, we introduce fairness properties that rely on the relationship between individuals and their group structure. By doing so, we seek to provide further insight into the intricacies of fairness in allocation problems involving groups of agents.

\subsection{Proportionally Envy-Free (PEF) Allocations} \label{sec:pef}
The first property we introduce, PEF, is a hybrid (and extension) of two existing notions of fairness -- individual \textit{proportionality} ($i$-PROP) \cite{brams1996fairdivision} in the fair division literature, and \wef{} introduced in Section 2. First, we restate the definition of a relaxed version of $i$-PROP.
\begin{definition}[Proportional up to one good]
    An allocation $\mathcal{A} = (A_1, \dots, A_n)$ is individually proportional up to one good ($i$-PROP1) if, for any agent $p_i \in N$, there exists a good $g \in G \setminus A_i$ such that $v_i(A_i \cup \{ g \}) \geq \frac{v_i(G)}{n}$.
\end{definition}

Next, we proceed to define PEF. A PEF allocation can be interpreted as a middle-ground between $i$-PROP and \wef{}. It mandates that every agent value their bundle as much as they value any other group's bundle, normalized by the group size. As usual, we introduce the ``up to one good'' relaxation of this notion.\footnote{We adopt a similar relaxation to the traditional $i$-PROP property in the literature, having the good \textit{g} added to the left-hand side of the equation rather than removing from the right-hand side.}

\begin{definition}[Proportionally envy-free up to one good]
    An allocation $\mathcal{A} = (A_1, \dots, A_n)$ is proportionally envy-free up to one good (PEF1) if, for any agent $p_i \in N$ and group $T_k \in \mathcal{T}$, there exists $g \in B_k \setminus A_i$ such that 
    $v_i(A_i \cup \{g\}) \geq \frac{v_i(B_k)}{w_k}$.
\end{definition}

It is known that \ef{1} implies $i$-PROP1 \cite{conitzer2017fairpublic}. Thus, a natural follow-up question would be whether \ef{1} implies PEF1, and it turns out that it is true. In fact, there is also a connection between PEF1 and $i$-PROP1, as the following proposition postulates. 

\begin{proposition}
    \ef{1} implies PEF1. Additionally, when all of the group sizes (and hence weights) are equal, PEF1 implies $i$-PROP1. 
\end{proposition}
\begin{proof}
    For the first part, we start by noting that from the definition of \ef{1}, for all $p_{i'} \in T_k$, there exists some $g_{i'} \in A_{i'}$ such that $v_i(A_i) \geq v_i(A_{i'}) - v_i(g_{i'})$. Summing both sides over agents $p_{i'} \in T_k$, we obtain
    \begin{equation}\label{pef-1}
        w_kv_i(A_i) \geq \sum_{i': p_{i'} \in T_k} [v_i(A_{i'}) - v_i(\{g_{i'}\})].
    \end{equation}
    The right-hand side can be simplified as follows, with $g_\text{max}$ being the maximally valued good by $p_i$ in the bundle $B_k \setminus A_i$:
    \begin{align*}
        & \sum_{i':p_{i'} \in T_k}  v_i(A_{i'}) - \sum_{i':p_{i'} \in T_k} v_i(\{g_{i'}\}) \\
        &= v_i(B_k) - \sum_{i':p_{i'} \in T_k} v_i(\{g_{i'}\}) \\
        &\geq v_i(B_k) - w_k v_i(g_\text{max})
    \end{align*}
    Combining this with (\ref{pef-1}), we get
    \begin{align} 
        & v_i(A_i) \geq \frac{v_i(B_k)}{w_k} - v_i(\{g_\text{max}\}) \nonumber \\ & \Rightarrow v_i(A_i \cup \{g_\text{max}\}) \geq \frac{v_i(B_k)}{w_k}. \label{pef-final}
    \end{align}
    Thus, PEF1 is satisfied.
    
    We now prove the second part of the proposition. Since the weights are equal, in this part, we write it as $w$. From the definition of PEF1, and summing both sides of (\ref{pef-final}) over all groups $T_k \in \mathcal{T}$ (recall that there are $\ell$ groups in total), we obtain
    \begin{equation}
        \ell \times v_i(A_i \cup \{g\}) \geq \sum_{k: T_k \in \mathcal{T}} \frac{v_i(B_k)}{w}.
    \end{equation}
    Hence, since $w\ell = n$, we have that
    \begin{align}
        v_i(A_i \cup \{g\}) & \geq \sum_{k: T_k \in \mathcal{T}} \frac{v_i(B_k)}{w\ell} \nonumber \\ 
        & = \frac{\sum_{k:T_k \in T} v_i(B_k)}{w\ell} = \frac{v_i(G)}{n}.
    \end{align}
    for some $g \in G \setminus A_i $. Thus, $i$-PROP1 is satisfied.
\end{proof}

As such, the SM-IWRR and IWRR algorithms proposed in section 3 naturally satisfies PEF1 (and $i$-PROP1 in the case of equal-size groups) in addition to the guarantees already shown.

\subsection{Approximately Group Stable Allocations} \label{sec:groupstability}
    The second property that we introduce is \textit{group stability}. There are scenarios whereby agents are able to declare a one-time membership to a group, and other instances where they can opt not to join any group at all, before the allocation process begins. This is in contrast to settings whereby agents inherently belong to certain groups, such as ethnic groups in housing allocation problems \cite{benabbou2018diversity}. We introduce the notion of group stability, and consider a relaxation of the concept, which we will term group \textit{$\epsilon$-stability} for use in our allocation problem. The significance of introducing such a notion is also exemplified in settings where the strategic reporting of membership to groups may result in undesirable effects. For instance, in the conference peer review setting, authors have the option to declare a track for the paper. This may invite strategic misreporting about the most appropriate track for the paper, in a bid to improve the chances of acceptance. We would like to introduce a notion that discourages this behaviour.
    
    One key thing to note here is that the notion of \textit{stability} here is implicit, in the sense that the agent will not be able to change their group membership after being in a group. However, an allocation satisfying such a property would have more merit as agents can be assured that they could not have been much better off by misreporting their preferences.
    
    An allocation mechanism $\mathcal{M} : N \times G \times \mathcal{T} \times V \rightarrow |N|^{G}$ is a function that takes in the set of agents, goods, group memberships, and valuations (where $V$ is the set of all agents' valuation functions), and outputs an allocation of goods to agents. We only consider deterministic allocation mechanisms, but the definitions can easily be extended to consider randomized ones as well.
    
    We now formally introduce the relaxed notion of the group stability property. 
    
    In fact, this relaxed notion is essentially an ``up to one good'' variation, and in many real-world settings, one could argue that a single good has little utility, thereby giving rise to an \textit{almost} stable property as defined below.

\begin{definition}[Group $\epsilon$-stability] \label{epsilon-stability}
    An allocation $\mathcal{A} = (A_1, \dots, A_n)$ returned by some mechanism $\mathcal{M}(N, G, \mathcal{T}, V)$ is group $\epsilon$-stable if the following conditions hold: 
    \begin{enumerate}[label=(\roman*)]
    \item For every agent $p_i \in N$, there exists some good $g \in A_i'$ such that
    $$v_i(A_i) \geq v_i(A'_i \setminus \{g\}),$$
    where $\mathcal{M}(N, G, \mathcal{T}', V) = \mathcal{A}' = (A_1', ..., A_n')$, and $\mathcal{T}'$ is equivalent to $\mathcal{T}$ with the difference being that $p_i$ is now in a group on its own.
    \item For every agent $p_i \in N$, and every group $T_k \in \mathcal{T}$, there exists some good $g \in A_i'$ such that
    $$v_i(A_i) \geq v_i(A_i^{(k)} \setminus \{g\}),$$
    where $\mathcal{M}(N, G, \mathcal{T}^{(k)}, V) = \mathcal{A}^{(k)} = (A_1^{(k)}, ..., A_n^{(k)})$, and $\mathcal{T}^{(k)}$ is defined by taking $\mathcal{T}$ and moving agent $p_i$ to group $T_k$.
    \end{enumerate}
\end{definition}

Intuitively, (i) caters for the case whereby agents are able to choose not to join a group prior to the allocation process. Then, the property guarantees that they will not have ``regretted'' their decision. (ii) is similar in this regard, but the ``no-regret'' is with respect to reporting membership to other groups instead.

The next question that arises is whether such a property is achievable. We give two theorems that provide a positive answer.

\begin{theorem} \label{thm-groupstable}
    The IWRR algorithm returns an allocation that is group $\epsilon$-stable.
\end{theorem}

\begin{proof}
    We first prove that IWRR returns an allocation that satisfies (ii) of the group $\epsilon$-stability property. Suppose that some agent $p_i \in N$ switches group from $T_k$ to $T_{k'}$. By the individual round-robin nature of the IWRR, every agent gets one good per round, regardless of their group. Let $g^{(d)}_{i}$ and $g'^{(d)}_{i}$ be the $d^\text{th}$ good (i.e, in round $r_d$) that $p_i$ received as part of being in $T_k$ and $T_{k'}$ respectively. Let there be a total of $K$ rounds. Then, since each agent selects their favourite good at every round, we must have that $v_i(\{g^{(d)}_{i}\}) \geq v_i(\{ g'^{(d+1)}_{i} \})$ for all $d = 1, \dots, K-1$. Thus, for all agents $p_i \in N$ belonging to group $T_k$, where $A_i$ is the bundle received by being in the group $T_k$, and $A'_i$ is the bundle received by declaring membership to any other group $T_{k'}$, we have that 
    \begin{equation*}
        \sum_{d=1}^{K-1} v_i(\{ g_i^{(d)} \}) + v_i(\{ g_i^{(K)} \}) \geq \sum_{d=2}^{K} v_i(\{ g_i'^{(d)} \})
    \end{equation*}
    where the left-hand side is equal to $v_i(A_i)$ and the right-hand side is equal to $v_i(A_i' \setminus \{{g'}_i^{(1)} \} )$, obtained by a relabelling of the index $d$. The proof that IWRR returns an allocation that satisfies (i) is similar to that of (ii), where we consider $T_{k'}$ to be an empty group initially, and if $p_i$ joins, then it becomes a singleton. The result follows.
\end{proof}
Given that the IWRR algorithm is group stable up to one good, we can say the same about the SM-IWRR algorithm, with the proof being a simple combination of Theorem \ref{thm-groupstable} and the representative good idea.
\begin{theorem}
    The SM-IWRR algorithm returns an allocation that is group $\epsilon$-stable.
\end{theorem}
In summary, we have shown that the SM-IWRR and IWRR algorithms also have group stability guarantees, further strengthening the fairness guarantees provided by these algorithms.

\end{document}